\newcommand{\cv}[1]{} % conference version 
\newcommand{\av}[1]{#1} % arxiv version

\cv{
\documentclass[twoside,leqno,twocolumn]{article}
\usepackage[letterpaper]{geometry}
}

\av{\documentclass[letter]{article}
  \usepackage[totalwidth=420pt,totalheight=625pt]{geometry}
\date{}
}

\cv{\usepackage{ltexpprt}}
\usepackage{hyperref}
\usepackage{times}
\usepackage{soul}
\usepackage{url}
\usepackage[utf8]{inputenc}
\usepackage{caption}
\usepackage{graphicx}
\usepackage{amsmath}
\av{\usepackage{amsthm}}
\usepackage{amssymb}
\usepackage{booktabs}
\usepackage{algorithmicx}
\usepackage{algorithm, algpseudocode}
\usepackage{setspace}
\usepackage{tikz,calc}

\av{
\newtheorem{theorem}{Theorem}[section]

\newtheorem{fact}{Fact}[section]

\newtheorem{proposition}{Proposition}[section]

}

\makeatletter
\cv{
\def\subsection{\@startsection{subsection}{2}{0pt}{-12pt}{5pt}{\normalsize\bf}}
}
\makeatother

\tikzstyle{N}=[fill=white, draw=black, shape=circle, inner
sep=0pt, minimum size=4pt]
\tikzstyle{C}=[fill=blue, draw=black, shape=circle, inner
sep=0pt, minimum size=4pt]
\tikzstyle{L}=[]

\tikzstyle{R}=[red,thick]

\tikzstyle{B}=[black,semithick]
\tikzstyle{BD}=[black,semithick, dashed]
\tikzstyle{GD}=[blue,dashed]
\tikzstyle{var}=[fill=white, draw=black, shape=circle, inner
sep=0pt, minimum size=4pt]
\tikzstyle{T}=[inner sep=1pt,minimum size=1pt]

\algnewcommand{\LineComment}[1]{\State \(\triangleright\) #1}

\newcommand{\tww}{\text{\normalfont\slshape tww}}
\newcommand{\lb}{\text{\normalfont\slshape lb}}
\newcommand{\ub}{\text{\normalfont\slshape ub}_{\mathrm{greedy}}}
\newcommand{\Prime}{\text{\normalfont\slshape prime}}
\newcommand{\Paley}[1]{\text{\normalfont Paley-}#1}
\newcommand{\pos}[1]{\phi_\prec(#1)}
\newcommand{\symdiff}{\mathop{\bigtriangleup}}

\newcommand{\Prec}{{\mathnormal{\prec}}}% to be used in (T,\Prec)

\providecommand{\shortcite}[1]{\cite{#1}}
\newcommand{\citet}[1]{\citeauthor{#1} \shortcite{#1}}

\newcommand{\SB}{\{\,}%
\newcommand{\SM}{\;{\mid}\;}%
\newcommand{\SE}{\,\}}%

\newcommand{\Card}[1]{|#1|}
\let\phi=\varphi
\let\epsilon=\varepsilon 
\def\hy{\hbox{-}\nobreak\hskip0pt}

\begin{document}

\cv{
\title{\Large A SAT Approach to Twin-Width}
\author{Andr\'{e} Schidler\thanks{Algorithms and
    Complexity Group, TU Wien, Vienna, Austria. The authors acknowledge the support from the Austrian Science Fund (FWF), projects P32441 and W1255, and from the WWTF, project ICT19-065.} \and Stefan Szeider\footnotemark[1]}
}
\av{
\title{A SAT Approach to Twin-Width\thanks{The authors acknowledge the support from the Austrian Science Fund (FWF), projects P32441 and W1255, and from the WWTF, project ICT19-065.}}
\author{%
Andr\'{e} Schidler and Stefan Szeider\\[4pt]
\small Algorithms and Complexity Group\\[-3pt]
\small TU Wien, Vienna, Austria\\[-3pt] 
\small \texttt{\{aschidler,sz\}@ac.tuwien.ac.at}
}
}

\date{}

\maketitle

% Default Copyright Statement
\cv{\fancyfoot[R]{\scriptsize{Copyright \textcopyright\ 2022 by SIAM\\
Unauthorized reproduction of this article is prohibited}}}

%\fancyfoot[R]{\scriptsize{Copyright \textcopyright\ 20XX\\
%Copyright for this paper is retained by authors}}

%\fancyfoot[R]{\scriptsize{Copyright \textcopyright\ 20XX\\
%Copyright retained by principal author's organization}}

\begin{abstract}\cv{ \small
  \begin{spacing}{1}}
    The graph invariant twin-width  was recently introduced by
    Bonnet, Kim, Thomass\'e, and Watrigan. Problems expressible in
    first-order logic, which includes many prominent NP-hard problems,
    are tractable on graphs of bounded twin-width if a certificate for
    the twin-width bound is provided as an input. Computing such a
    certificate, however, is an intrinsic problem, for which no
    nontrivial algorithm is known.
 
    In this paper, we propose the first practical approach for
    computing the twin-width of graphs together with the corresponding
    certificate. We propose efficient SAT-encodings that rely on a
    characterization of twin-width based on elimination sequences.
    This allows us to determine the twin-width of many famous graphs
    with previously unknown twin-width.  We utilize our encodings to
    identify the smallest graphs for a given twin-width bound
    $d \in \{1,\dots,4\}$.  \cv{ \end{spacing}}
\end{abstract}

\section{Introduction}
Twin-width is a new graph invariant that was recently introduced by
Bonnet \emph{et al.}~\shortcite{BonnetEtal(2)21,BonnetEtal(3)20,BonnetEtal(1)20}, inspired
by previous work by Guillemot and Marx~\cite{GuillemotMarx14}.  Graph classes of bounded
twin-width admit the fixed-parameter tractability of First-Order (FO)
model checking, parameterized by the length of the FO formula,
provided a witness for bounded twin-width is given.  Many NP-hard
problems such as as ``does the input graph contain an independent set
of size at least $r$?'' or ``does the input graph contain a subgraph
that is isomorphic to a fixed pattern graph $H$?''  can be naturally
expressed as FO model checking.  Graph classes of bounded twin-width
subsume and generalize several dense graph classes for which FO model
checking is fixed-parameter tractable, including map graphs, bounded
rank-width graphs, bounded clique-width graphs, cographs, and unit
interval graphs.  Thus, twin-width boundedness plays a similar role
for dense graph classes as \emph{nowhere density} plays for sparse
graph classes~\cite{GroheKreutzerSiebertz17}.

Bonnet \emph{et al.}'s \shortcite{BonnetEtal(1)20} FO model checking
algorithm for graphs of bounded twin-width requires a certificate that
the input graph's twin-width is bounded by a constant $d$.  The most
pressing open theoretical question regarding twin-width concerns the
complexity of computing such a certificate, and more generally,
recognize graphs of twin-width $\leq d$~\cite{BonnetEtal(1)20}. There
are no practical algorithms known to compute the twin-width of a graph
exactly or approximately.

\subsection{Contribution}
 In this paper, we take a SAT-based  approach to the exact computation of
twin-width. We thereby utilize the power of SAT solving (solving the
propositional satisfiability problem SAT) for a
combinatorial problem, continuing a compelling and successful line of
research~\cite{CodishFrankItzhakovMiller16,CodishMillerProsserStuckey19,HeuleSzeider15,Heule18,HeuleKullmann17,PeitlSzeider21b,SchidlerSzeider21b}. As a result, we can identify the exact twin-width
of many graphs for which the twin-width was previously unknown.  

More specifically, we propose two SAT encodings that take a graph $G$
and an integer $d$ as input, and produce a propositional CNF formula
$F(G,d)$, which is satisfiable if and only if the twin-width of $G$ is
at most~$d$. By running a SAT-solver on $F(G,d)$ for different values
of~$d$, we can determine the exact twin-width of~$G$. We propose
methods for computing lower and upper bounds for $d$ that allow us to
reduce the interval of possible values of $d$ for running the SAT
solver on. Both encodings are based on a new characterization of
twin-width in terms of elimination orderings, which are somewhat
related to SAT encodings used for other width
measures~\cite{GanianLodhaOrdyniakSzeider19,SamerVeith09,SchidlerSzeider20}. However,
for twin-width, the situation is more involved, because it is not
sufficient to globally bound certain static values (like out-degrees
in an elimination ordering for treewidth~\cite{SamerVeith09}).

We demonstrate the potentials and limits of our encodings by utilizing
them in the following three computational experiments.

\begin{enumerate}
\item \emph{Twin-width of small Random Graphs}.  We determine
  experimentally how the twin-width of a random graph depends on its
  density. As one expects, the  twin-width is small for dense and
  sparse graphs. Graphs of  edge-probability 0.5 have the highest
  twin-width.

\item  \emph{Twin-width of Famous Named Graphs}. Over many decades of
  research in combinatorics, researchers have collected several
  special graphs, which have been used as counterexamples for
  conjectures or for showing the tightness of combinatorial
  results. We considered several of such special graphs from the
  literature and computed their exact twin-width. We believe
  that these results will be of interest to people working in
  combinatorics. This way, we have identified a certain class of
  strongly regular graphs (Paley graphs) that provide high lower
  bounds for twin-width.

\item \emph{Twin-Width Numbers}. In general, it is not known how many
  vertices are required to form a graph of a certain twin-width. In
  fact, there is limited knowledge on lower-bound techniques for
  twin-width.  We use our SAT encoding together with a graph generator
  to identify the smallest graphs of twin-width $1,2,3,4$, and provide
  tight bounds for twin-width $5$ and $6$. This way, we can determine
  the first few twin-width numbers, where the $d$-th twin-width number
  is the smallest number of vertices of a graph with twin-width~$d$.
  A similar computation has been conducted for
  clique-width~\cite{HeuleSzeider15}.  Interestingly, up to
  isomorphism, there are unique smallest graphs of twin-width 1, 2,
  and 4, respectively, and there are five such graphs for
  twin-width~3.
\end{enumerate}

\section{Twin-width}\label{section:twin-width}

A \emph{trigraph} is an undirected graph $G$ with vertex set $V(G)$ whose
edge set $E(G)$ is partitioned into a set $B(G)$ of black edges and a
set $R(G)$ of red edges. We consider an ordinary graph as a trigraph
with all its edges being black.
The set $N_G(v)$ of neighbors of a vertex $v$ in a trigraph~$G$
consists of all the vertices adjacent to $v$ by a black or red edge.
We call $u\in N_G(v)$ a \emph{black neighbor} of $v$ if $uv\in B(G)$
and we call it a \emph{red neighbor} if $uv\in R(G)$.
The \emph{red degree} of a vertex $v\in V(G)$ of a trigraph $G$ is the number of
its red neighbors.  A $d$-trigraph is a trigraph where each vertex has
red degree at most $d$.

\subsection{Twin-Width via Sequences of $d$-Contractions}
We give the original definition of twin-width~\cite{BonnetEtal(2)21,BonnetEtal(3)20,BonnetEtal(1)20}.

A trigraph $G'$ is obtained from a trigraph
$G$ by \emph{contraction}: two (not-necessarily adjacent)
vertices $u$ and $v$ are merged into a single vertex $w$, and the edges
of $G$ are updated as follows:
Every vertex in the \emph{symmetric difference} $N_G(u) \symdiff N_G(v)$ is
made a red neighbor of  $w$.
If a vertex $x\in N_G(u) \cap N_G(v)$ is a black neighbor of both $u$ and $v$, then $w$ is made
a black neighbor of $x$; otherwise, $w$ is made a red neighbor of $x$.
The other edges  (not incident with $u$ or~$v$) remain unchanged.

\sloppypar A \emph{sequence of $d$-contractions} or \emph{$d$-sequence} for a
graph~$G$ is a sequence of $d$-trigraphs
$G_0$, $G_1, \dots, G_{n-1}$ where $G_0 = G$, $G_{n-1}$ is the graph on a
single vertex, and $G_{i}$ for $i\geq 1$  is obtained from $G_{i-1}$ by contraction.
We observe that $\Card{V(G_i)}=n-i$ for $0\leq i < n =\Card{V(G)}$.
The \emph{twin-width} of a trigraph $G$, denoted $\tww(G)$, is the
smallest integer $d$ such that $G$ admits a $d$-sequence.

It is indeed sometimes  necessary to contract
non-adjacent vertices. For instance, 
Figure~\ref{fig:WagnerContractions} shows a sequence of 2-contractions
for the Wagner graph.  Without contracting non-adjacent vertices,  a vertex of red degree
$>2$ would be created by the first contraction since each vertex has degree~3 and shares no neighbor with any of its neighbors.

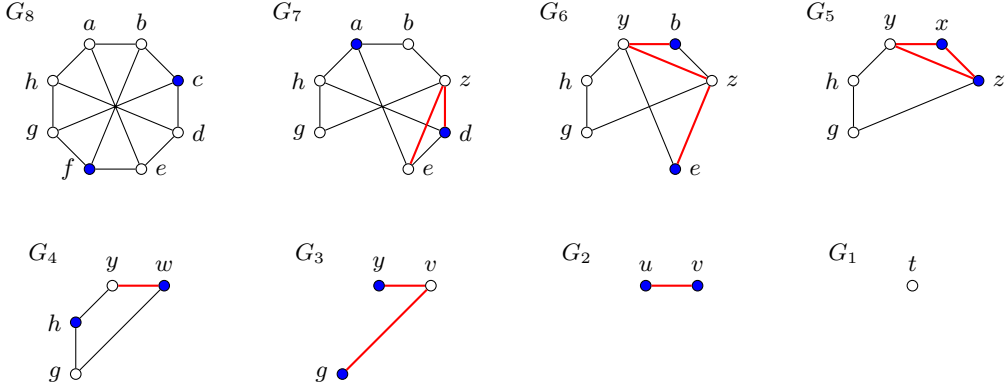
\begin{figure}[tbh]
  \centering
  \tikzset{
    lead/.style={
        execute at end picture={
            \coordinate (lower right) at (current bounding box.south east);
            \coordinate (upper left) at (current bounding box.north west);
        }
    },
    follow/.style={
        execute at end picture={
            \pgfresetboundingbox
            \path (upper left) rectangle (lower right);
        }
    }
}

\begin{tikzpicture}[scale=0.9, lead]\small % 1
    \draw
%    (0,1.8) node[coordinate] (lowpoint) {}
    (-1.4,1.4) node (number) {$G_8$}
    (0,0) node[coordinate] (origin) {}
    +(-5.5*360/8: 1cm) node[N, label={$a$}] (1) {}
    +(-6.5*360/8: 1cm) node[N, label={$b$}] (2) {}
    +(-7.5*360/8: 1cm) node[C, label=right:{$c$}] (3) {}
    +(-0.5*360/8: 1cm) node[N, label=right:{$d$}] (4) {}
    +(-1.5*360/8: 1cm) node[N, label=right:{$\strut e$}] (5) {}
    +(-2.5*360/8: 1cm) node[C, label=left:{$f$}] (6) {}
    +(-3.5*360/8: 1cm) node[N, label=left:{$g$}] (7) {}
    +(-4.5*360/8: 1cm) node[N, label=left:{$h$}] (8) {}
    ;
    \draw 
    (1)--(2)--(3)--(4)--(5)--(6)--(7)--(8)--(1);
    \draw (1)--(5) (2)--(6) (3)--(7) (4)--(8);
\end{tikzpicture}\hspace{6mm}
\begin{tikzpicture}[scale=0.9, lead]\small % 2
  \draw
  (-1.4,1.4) node (number) {$G_7$}
  (0,1.8) node[coordinate] (lowpoint) {}
  (0,0) node[coordinate] (origin) {}
  +(-5.5*360/8: 1cm) node[C, label={$a$}] (1) {}
  +(-6.5*360/8: 1cm) node[N, label={$b$}] (2) {}
  +(-7.5*360/8: 1cm) node[N, label=right:{$z$}] (36) {}
  +(-0.5*360/8: 1cm) node[C, label=right:{$d$}] (4) {}
  +(-1.5*360/8: 1cm) node[N, label=right:{$\strut e$}] (5) {}
  +(-3.5*360/8: 1cm) node[N, label=left:{$g$}] (7) {}
  +(-4.5*360/8: 1cm) node[N, label=left:{$h$}] (8) {}
  ;
  \draw (1)--(2) (4)--(5) (7)--(8)--(1)--(5)  (4)--(8);
  \draw (36)--(2) (36)--(7);
  \draw[R] (36)--(5) (36)--(4); 
\end{tikzpicture}\hspace{6mm}
\begin{tikzpicture}[scale=0.9, lead]\small % 3
  \draw
  (-1.4,1.4) node (number) {$G_6$}
  (0,1.8) node[coordinate] (lowpoint) {}
  (0,0) node[coordinate] (origin) {}
  +(-5.5*360/8: 1cm) node[N, label={$y$}] (14) {}
  +(-6.5*360/8: 1cm) node[C, label={$b$}] (2) {}
  +(-7.5*360/8: 1cm) node[N, label=right:{$z$}] (36) {}
  +(-1.5*360/8: 1cm) node[C, label=right:{$\strut e$}] (5) {}
  +(-3.5*360/8: 1cm) node[N, label=left:{$g$}] (7) {}
  +(-4.5*360/8: 1cm) node[N, label=left:{$h$}] (8) {}
  ;
  \draw  (7)--(8)  (36)--(2) (36)--(7);
  \draw  (8)--(14)--(5);
  \draw[R] (2)--(14)--(36)--(5); 
\end{tikzpicture}\hspace{6mm}
\begin{tikzpicture}[scale=0.9, follow]\small % 4
  \draw (-1.4,1.4) node (number) {$G_5$}
  (0,1.8) node[coordinate] (lowpoint) {}
  (0,0) node[coordinate] (origin) {}
  +(-5.5*360/8: 1cm) node[N, label={$y$}] (14) {}
  +(-6.5*360/8: 1cm) node[C, label={$x$}] (25) {}
  +(-7.5*360/8: 1cm) node[C, label=right:{$z$}] (36) {}
  +(-3.5*360/8: 1cm) node[N, label=left:{$g$}] () {}
  +(-4.5*360/8: 1cm) node[N, label=left:{$h$}] (8) {}
  ;
  \draw  (7)--(8)   (36)--(7);
  \draw  (8)--(14);
  \draw[R] (25)--(36)--(14)--(25); 
\end{tikzpicture}\hspace{6mm}\av{

  \bigskip
  }
\begin{tikzpicture}[scale=0.9, follow]\small % 5
  \draw      (-1.4,1.4) node (number) {$G_4$}
  (0,1.8) node[coordinate] (lowpoint) {}
  (0,0) node[coordinate] (origin) {}
  +(-5.5*360/8: 1cm) node[N, label={$y$}] (14) {}
  +(-6.5*360/8: 1cm) node[C, label={$w$}] (2356) {}
  +(-3.5*360/8: 1cm) node[N, label=left:{$g$}] (7) {}
  +(-4.5*360/8: 1cm) node[C, label=left:{$h$}] (8) {}
  ;
  \draw  (14)--(8)--(7)--(2356);  
  \draw[R] (14)--(2356);
\end{tikzpicture}\hspace{6mm}
\begin{tikzpicture}[scale=0.9, follow]\small % 6
  \draw      (-1.4,1.4) node (number) {$G_3$}
  (0,1.8) node[coordinate] (lowpoint) {}
  (0,0) node[coordinate] (origin) {}
  +(-5.5*360/8: 1cm) node[C, label={$y$}] (14) {}
  +(-6.5*360/8: 1cm) node[N, label={$v$}] (23568) {}
  +(-3.5*360/8: 1cm) node[C, label=left:{$g$}] (7) {}
  ;
  \draw[R] (14)--(23568)--(7);
\end{tikzpicture}\hspace{6mm}
\begin{tikzpicture}[scale=0.9, follow]\small % 7
  \draw      (-1.4,1.4) node (number) {$G_2$}
 % (0,1.8) node[coordinate] (lowpoint) {}
  (0,0) node[coordinate] (origin) {}
  +(-5.5*360/8: 1cm) node[C, label={$u$}] (147) {}
  +(-6.5*360/8: 1cm) node[C, label={$v$}] (23568) {}
  ;
  \draw[R] (147)--(23568);
\end{tikzpicture}\hspace{6mm}
\begin{tikzpicture}[scale=0.9, follow]\small % 8
  \draw      (-1.4,1.4) node (number) {$G_1$}
%  (0,1.8) node[coordinate] (lowpoint) {}
  (0,0) node[coordinate] (origin) {}
  +(-5.5*360/8: 1cm) node[N, label={$t$}] (12345678) {}
  ;
\end{tikzpicture}%
\caption{A sequence of 2-contractions for the Wagner graph. Vertices that
  will be contracted next are marked blue.}
  \label{fig:WagnerContractions}
\end{figure}

We state here some basic properties of twin-width, observed in the
original paper~\cite{BonnetEtal(1)20}.

\begin{fact}\label{fact:subgraph}
  If $G'$ is and induced  subgraph of a graph $G$, then
  $\tww(G')\leq \tww(G)$. 
\end{fact}
 
For a graph $G$, we denote by $\overline{G}$ its \emph{complement
  graph}, which is defined by $V(\overline{G})=V(G)$ and
$E(\overline G)=\SB uv \SM u,v\in V(G), uv \notin E(G), u \neq v \SE$.
\begin{fact}\label{fact:complement}
  For every graph $G$, we have $\tww(G)=\tww(\overline{G})$.
\end{fact}

\subsection{Twin-Width via $d$-Elimination Sequences}\label{sec:elim}

Next we give an alternative definition of twin-width which is better
suited for formulating our  SAT encodings.

Let $G$ be a graph, $T$ a tree with $V(T)=V(G)$, rooted at some
vertex $r_T$, and $\prec$ a linear ordering of $V(T)$,
where $u \prec v$ for two vertices $u, v\in V(T)$ such that $v$ is the parent of $u$ in $T$. 
We call $T$ a \emph{contraction tree}, $\prec$ an \emph{elimination ordering},
and the pair $(T,\Prec)$ a \emph{twin-width decomposition} of $G$.
Thus, when we write $V(G)=\{v_1,\dots,v_n\}$ such that  $v_1 \prec \dots
\prec v_n$ and $v_n = r_T$, then 
% For $v\in V(T)\setminus \{r_T\}$ we write $p_T(v)$ for the parent of
% $v$ in $T$.
$T$ and $G$ define a sequence of graphs $H_0,\dots,H_{n-1}$
with $V(H_i)=\{v_{i+1},\dots,v_{n}\}$.
%i.e., $H_0=G$ and $H_n$ is the
%empty graph.
We denote by $p_{i}$ the parent of $v_{i}$ in $T$. By
definition, $v_{i} \prec p_{i}$.

We define the edge set $E(H_{i})$ recursively as follows.  For $i=0$,
we set $E(H_0)=\emptyset$, and for $1 \leq i < n$, we set
\begin{subequations}
\begin{align}
   &E(H_i)
  =\SB uv\in E(H_{i-1}) \SM u,v \in V(H_i) \SE \label{eq:keep}\\
  &~~\av{~~~~~~~~~~~~~~~}\cup \; \SB up_{i} \SM  v_{i}u \in E(H_{i-1}) \SE \label{eq:transfer}\\
  &~~\av{~~~~~~~~~~~~~~~}\cup \; \SB up_{i} \SM  v_{i}u \in E(G), p_{i}u \notin E(G), u \in V(H_i)  \label{eq:create-source}\SE
       \\
  &~~\av{~~~~~~~~~~~~~~~}\cup \; \SB up_{i} \SM  v_{i}u \notin E(G), p_{i}u \in E(G), u \in V(H_i)
      \SE \label{eq:create-target}.
\end{align}
\label{eqn:all-lines}
\end{subequations}
%
%\[
%\begin{array}{l@{~}c@{~}ll}
%  E(H_i)
%  &=& \SB uv\in E(H_{i-1}) \SM u,v \in V(H_i) \SE  & (1)\\
%  &&\cup \; \SB up_i \SM  v_{i-1}u \in E(H_{i-1}) \SE & (2)\\
%  &&\cup \; \SB up_i \SM  v_{i-1}u \in E(G) \wedge p_{i-1}u \notin E(G)  \SE
%       & (3)\\
%  &&\cup \; \SB up_i \SM  v_{i-1}u \notin E(G) \wedge p_{i-1}u \in E(G)
%      \SE. & (4)
%\end{array}
%\]
%
\noindent We call the sequence $H_0,\dots,H_{n-1}$ the
\emph{elimination sequence} for $G$ defined by the twin-width
decomposition $(T,\Prec)$; if for an integer $d$, all the
$H_i$ have a maximum degree $\leq d$, we call $H_0,\dots,H_{n-1}$ a
$d$\hy elimination sequence. The \emph{width} of the twin-width
decomposition $(T,\Prec)$ of $G$ is the smallest integer $d$ such
that $(T,\Prec)$ defines a $d$\hy elimination sequence.

Figure~\ref{fig:WagnerELO} shows an example of a 2-elimination
sequence, and in Figure~\ref{fig:wagnerTree} the same elimination
sequence is superimposed on the graph.

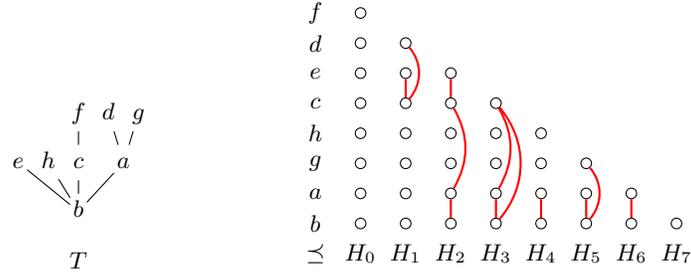
\begin{figure}[tbh]
  \centering
    \begin{tikzpicture}[xscale=0.4,yscale=0.65 ]\small % 1
  \draw
  (0,-1) node[T] (t) {$T$}
  (0,0) node[T] (b) {$\strut b$}
  (0,1) node[T] (c) {$\strut c$}
  (0,2) node[T] (f) {$\strut f$}
   (-1,1) node[T] (h) {$\strut h$}
  (-2,1) node[T] (e) {$\strut e$}  
  (1.5,1) node[T] (a) {$\strut a$}
  (1,2) node[T] (d) {$\strut d$}  
  (2,2) node[T] (g) {$\strut g$}
  ;
  \draw (a)--(b)--(c)--(f) (e)--(b)--(h) (d)--(a)--(g);
\end{tikzpicture}%
\hspace{6mm}\av{\hspace{14mm}}%
  \begin{tikzpicture}[scale=0.4]\small % 1
  \begin{scope}[xshift=-15mm]
    \draw
    (0,0) node (3) {$f$}
    (0,-1) node (1) {$d$}
    (0,-2) node (2) {$e$}

    (0,-3) node (5) {$c$}
    (0,-4)  node (8) {$h$}
    (0,-5) node (6) {$g$}
    (0,-6)  node (4) {$a$}
    (0,-7) node (7) {$b$}
    (0,-8) node (h) {$\preceq$}
    ;
  \end{scope}
  \draw
  (0,0) node[N] (3) {}
  (0,-1) node[N] (1) {}
  (0,-2) node[N] (2) {}
  (0,-3) node[N] (5) {}
  (0,-4)  node[N] (8) {}
  (0,-5) node[N] (6) {}
  (0,-6)  node[N] (4) {}
  (0,-7) node[N] (7) {}
  (0,-8) node (7) {$H_0$}
  ;
  \begin{scope}[xshift=15mm]
    \draw
    (0,-1) node[N] (1) {}
    (0,-2) node[N] (2) {}
    (0,-3) node[N] (5) {}
    (0,-4)  node[N] (8) {}
    (0,-5) node[N] (6) {}
    (0,-6)  node[N] (4) {}
    (0,-7) node[N] (7) {}
    (0,-8) node (h) {$H_1$}
    ;
    \draw[R] (1) edge [bend left=40] (5);
    \draw[R] (2)--(5); 
  \end{scope}

  \begin{scope}[xshift=30mm]
    \draw
    (0,-2) node[N] (2) {}
    (0,-3) node[N] (5) {}
    (0,-4)  node[N] (8) {}
    (0,-5) node[N] (6) {}
    (0,-6)  node[N] (4) {}
    (0,-7) node[N] (7) {}
    (0,-8) node (h) {$H_2$}
    ;

    \draw[R] (2) -- (5);
    \draw[R] (5) edge [bend left=30] (4);
    \draw[R] (4)--(7);
  \end{scope}

  \begin{scope}[xshift=45mm]   
    \draw
    (0,-3) node[N] (5) {}
    (0,-4)  node[N] (8) {}
    (0,-5) node[N] (6) {}
    (0,-6)  node[N] (4) {}
    (0,-7) node[N] (7) {}
    (0,-8) node (h) {$H_3$};

    \draw[R] (5) edge [bend left=30] (4);
    \draw[R] (5) edge [bend left=40] (7);
    \draw[R] (4) -- (7);

  \end{scope}
  \begin{scope}[xshift=60mm]   
    \draw
    (0,-4)  node[N] (8) {}
    (0,-5) node[N] (6) {}
    (0,-6)  node[N] (4) {}
    (0,-7) node[N] (7) {}
    (0,-8) node (h) {$H_4$}
    ;
    \draw[R] (7)--(4);
  \end{scope}
  \begin{scope}[xshift=75mm]   
    \draw
    (0,-5) node[N] (6) {}
    (0,-6)  node[N] (4) {}
    (0,-7) node[N] (7) {}
    (0,-8) node (h) {$H_5$}
    ;

    \draw[R] (7)--(4);
    \draw[R] (6) edge [bend left=40] (7);
  \end{scope}
  \begin{scope}[xshift=90mm]   
    \draw
    (0,-6)  node[N] (4) {}
    (0,-7) node[N] (7) {}
    (0,-8) node (h) {$H_6$}
    ;
    \draw[R] (4)--(7);
  \end{scope}
  \begin{scope}[xshift=105mm]   
    \draw
    (0,-7) node[N] (7) {}
    (0,-8) node (h) {$H_7$}
    ;
  \end{scope}
\end{tikzpicture}%

%%% Local Variables:
%%% mode: latex
%%% TeX-master: "twwsat"
%%% End:
    \caption{A 2-elimination sequence for the Wagner graph, defined by
      the linear ordering $\prec$ and the contraction tree
      $T$. This is the 2-elimination sequence that we get by
        applying the construction from the proof of
        Theorem~\ref{the:leo} to the sequence of 2-contractions shown
        in Figure~\ref{fig:WagnerContractions}.}
\label{fig:WagnerELO}
\end{figure}

\begin{figure}[tbh]
\centering
  \begin{tikzpicture}[scale=1]\small % 1
    \draw
%    (0,1.8) node[coordinate] (lowpoint) {}
    (0,0) node[coordinate] (origin) {}
    +(-5.5*360/8: 1cm) node[N, label={$7$}] (7) {}
    +(-6.5*360/8: 1cm) node[N, label={$8$}] (8) {}
    +(-7.5*360/8: 1cm) node[N, label=right:{$4$}] (4) {}
    +(-0.5*360/8: 1cm) node[N, label=right:{$2$}] (2) {}
    +(-1.5*360/8: 1cm) node[N, label=below:{$3$}] (3) {}
    +(-2.5*360/8: 1cm) node[N, label=below:{$1$}] (1) {}
    +(-3.5*360/8: 1cm) node[N, label=left:{$6$}] (6) {}
    +(-4.5*360/8: 1cm) node[N, label=left:{$5$}] (5) {}
    ;
    \draw[style=GD] (1)--(4) (2)--(7) (6)--(7) (5)--(8) (3)--(8);
    \draw[style=BD] (4)--(8)--(7);
    \draw[style=B]  (4)--(2)--(3)--(1)--(6)--(5)--(7)
    (2)--(5) (4)--(6) (7)--(3) (8)--(1);
  \end{tikzpicture}
  \caption{The Wagner graph with linear ordering $\prec$ from
    Figure~\ref{fig:WagnerELO} indicated by index numbers. The
    contraction tree~$T$ is superimposed on the graph, where blue
    dashed edges indicate tree edges that are not shared with the
    graph, and black dashed edges indicate tree edges that are shared
    with the graph.}
  \label{fig:wagnerTree}
\end{figure}
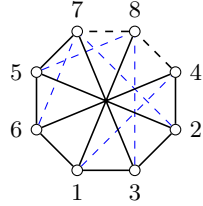

\begin{theorem}\label{the:leo}
  Let $G$ be a graph and $<$ an arbitrary linear ordering of $V(G)$.  $G$
  has twin-width $\leq d$ if and only if there exists a twin-width
  decomposition $(T,\Prec)$ of width $\leq d$ such that
  \begin{enumerate}
  \item if $x$ is the parent of $y$ in $T$, then
    $x < y$;
  \item the root of $T$ is the $<$-maximal element of  $V(G)$.
  \end{enumerate}
\end{theorem}
\begin{proof}
  Let $G$ be a graph and assume that $\tww(G)\leq d$.  By definition,
  there exists a $d$-sequence \cv{$G_0, G_1 , \dots, G_{n-1}$}\av{$G_0$, $G_1 , \dots, G_{n-1}$}, and each
  $G_i$, $i>0$, is obtained from $G_{i-1}$ by contracting two vertices
  $u_i$ and $v_i$, i.e., merging them into $w_i$, a new
  vertex. We slightly change contraction steps. Instead of introducing
  a new vertex $w_i$, we reuse one of the two vertices $u_i,v_i$ as
  $w_i$.  We use the ordering $<$ to decide which of the two vertices
  to reuse:
   \begin{equation}
     \label{eq:choise}
     w_i=\begin{cases}
       u_i & \text{if $u_i > v_i$},\\
       v_i & \text{otherwise}.
     \end{cases}
   \end{equation}
   This way, we obtain  a sequence $G_0', G_1' , \dots, G_{n-1}'$, with
   $V(G_i')\subseteq V(G)$, where each $G_i'$ is isomorphic to $G_i$.
   Since $V(G)=V(G_0') \supsetneq\dots\supsetneq V(G_{n-1}')$, this
   gives us a linear ordering $\prec$ of $V(G)$ in a natural way.  We
   obtain a contraction tree $T$ by taking $V(T)=V(G)$ and
   $E(T)=\SB u_iv_i\SM 1\leq i \leq n-1\SE$.  Because of~\eqref{eq:choise},
  the contraction tree satisfies the two
   conditions claimed in the statement of the theorem.  A $d$\hy
   elimination sequence $H_0,\dots,H_{n-1}$ is provided by taking
   $H_i$ as the subgraph of $G_i'$ formed by its red edges. Thus
   $(T,\Prec)$ is a twin-width decomposition of $G$ of width $\leq d$.

   Conversely, assume $(T,\Prec)$ is a twin-width decomposition of $G$
   of width $\leq d$. Let $H_0,\dots,H_{n-1}$ be the corresponding
   $d$\hy elimination sequence. We  turn the $d$\hy elimination
   sequence into a $d$-sequence by contracting pairs of vertices as
   indicated by~$T$. Hence $\tww(G)\leq d$.
 \end{proof}

\section{Preprocessing}\label{section:preprocessing}
In this section, we show how to decompose a given graph $G$ in
polynomial time into a collection $\Prime(G)$ of induced subgraphs of
$G$, such that $\tww(G)=\max_{H\in \Prime(G)}\tww(H)$. This
  decomposition can serve as a preprocessing step for  twin-width
  computation.

We require some
definitions. A \emph{module} of a graph $G$ is a nonempty set
$M\subseteq V(G)$ such that for any $x,y\in M$ and
$z\in V(G)\setminus M$ we have $xz\in E(G)$ if and only if
$yz\in E(G)$. A module $M$ is \emph{trivial} if $M=V(G)$ or
$\Card{M}=1$.  $M$ is a maximal module if it is not strictly contained
in any nontrivial module.  A graph is \emph{prime} if all its maximal
modules are trivial.  For every graph $G$, there exists a unique
partition $P_{\max}$ of $V(G)$ into maximal modules $M_1,\dots,M_s$,
and this partition can be found in linear time~\cite{CournierHabib94,McconnellSpinrad94}.  This partition gives rise
to the quotient graph $G/P_{\max}$ whose vertices are the maximal
modules of~$P$, and where two modules $M_i,M_j$, $i\neq j$, are joint
by an edge if and only if all the pairs $x_i\in M_i, x_j\in M_j$ are
joined by an edge in $G$. If we select for each module $M_i$ a
representative vertex $x_i\in M_i$, then the set $\{x_1,\dots,x_s\}$
of representatives induces a subgraph of $G$ that is isomorphic to
$G/P_{\max}$.  If $G$ and its complement graph $\overline{G}$ are
connected, then $G/P_{\max}$ is a prime graph~\cite{Gallai67,HabibPaul10}.  We recursively define the set
$\Prime(G)$ as follows:
\begin{enumerate}
\item If $G$ is disconnected, then $\Prime(G)$ is the union of the
  sets $\Prime(C)$ for all connected components $C$ of~$G$.
\item If $\overline G$ is disconnected, then  $\Prime(G)$ is the union of the
  sets $\Prime(\overline C)$ for all connected components $C$ of~$\overline G$.
\item If both $G$ and $\overline G$ are connected, then $\Prime(G)$ is the
  union of $\{G/P_{\max}\}$ and the sets $\Prime(G[M])$ for all
  nontrivial $M\in P_{\max}$.   
\end{enumerate}

\noindent The three cases above give rise to the \emph{modular decomposition}
of the graph $G$, represented as a rooted tree~\cite{HabibPaul10}.
The root of the tree is associated with $G$, the children of each vertex
are associated with the connected components (cases 1 and 2), or the
maximal modules (case 3) of the graph associated with their
parent. The leaves of the tree are in a 1-to-1 correspondence with the 
vertices of~$G$. 
\begin{theorem}\label{the:moddec}
  For every graph $G$ we have $\tww(G)=\max_{P\in \Prime(G)} \tww(P)$.
\end{theorem}
\begin{proof}  
  Let $d=\max_{P\in \Prime(G)} \tww(P)$.  As observed above,
  $G/P_{\max}$ is isomorphic to an induced subgraph of $G$; by
  induction, this holds for all the graphs in $\Prime(G)$.  Because of
  Fact~\ref{fact:subgraph}, $\tww(G) \geq d$ follows.  

  For showing $\tww(P) \leq d$, we proceed by induction on
  $\Card{V(G)}=n$.  The statement is certainly true if $n=1$, since
  then $\Prime(G)=\{G\}$.  Now assume $n>1$. We distinguish several
  cases.

  Consider the case where $G$ is disconnected into components
  $C_1,\dots,C_r$. For each $1\leq i \leq r$ we have
  $\Prime(C_i)\subseteq \Prime(G)$, and so, by induction, we have
  \mbox{$\tww(C_i)\leq \max_{P\in \Prime(C_i)} \tww(P) \leq d$}. Thus, for
  each $C_i$ there is a $d$\hy sequence ending in a single-vertex
  graph. Using the contractions of these $d$\hy sequences we obtain a
  $d$\hy sequence for $G$, which ends in an edgeless graph that
  consists of $r$ isolated vertices. We can extend this $d$\hy
  sequence by contracting the isolated vertices pairwise in any order,
  obtaining eventually  a single-vertex graph, without generating any
  red edges.  Thus $\tww(G)\leq d$. The case where $\overline{G}$ is disconnected
  follows from the previous argument and Fact~\ref{fact:complement}.

  Finally, assume that $G$ and $\overline{G}$ are connected. Thus
  $G/P_{\max}$ is prime and is isomorphic to an induced subgraph
  $G'\in \Prime(G)$ of $G$. For each $M \in P_{\max}$,
  $\Prime(G[M])\subseteq \Prime(G)$. By induction hypothesis,
  $\tww(G') \leq d$ and $\tww(G[M]) \leq d$. We thus obtain a $d$\hy
  sequence for $G$ by putting together $d$\hy sequences for $G[M]$,
  $M\in P_{\max}$, and a $d$\hy sequence for $G'$, which contract
  first each $G[M]$ on a single vertex of $G'$, and then contract $G'$
  on a single vertex.  Hence $\tww(G)\leq d$.
\end{proof}

Theorem~\ref{the:moddec} provides the basis for a preprocessing phase
for twin-width computation. If the given graph $G$ is not prime, we
compute $\Prime(G)$ and determine the twin-width of all the graphs in
$\Prime(G)$. Since for a non-prime graph $G$, the graphs in
$\Prime(G)$ are smaller than $G$, it is more efficient to run a costly
twin-width algorithm on the the graphs in $\Prime(G)$ than on $G$
itself. Hence, the preprocessing can be highly beneficial for
non-prime graphs.

\section{SAT Encodings}
In this section, we present two SAT encodings for twin-width.
Assume, we are given a graph $G$ with vertices $v_1 \dots v_n$ and
an integer $d$. We will define a propositional formula $F(G, d)$ in
Conjunctive Normal Form (CNF) that is satisfiable if and only if
$\tww(G)\leq d$.  For the construction of $F(G,d)$, we use the
characterization of twin-width in terms of a twin-width decomposition
$(T, \Prec)$, as established in Theorem~\ref{the:leo}.  We use the
indices $1 \leq i,j,k,m \leq n$ and subsequently omit the upper and
lower bounds for readability.  Furthermore, we use the mapping
$\pos{v_i}$ to denote the position of $v_i$ in $\prec$. We give two
different encodings for $F(G, d)$.  \newcommand{\X}{\hspace{20mm}}

\subsection{Relative Encoding}

\begin{table*}[t]
  \centering
  \caption{The variables used in the relative encoding.}
  \label{tab:var-relative}
\begin{tabular}{@{}l@{\X}c@{\X}l@{}}
     \toprule
     Name & Range & Meaning \\
      \midrule
     $a_{i,j}$    & $1 \leq i < j \leq n$             & $v_i v_j \in E_k$ for some $k$ \\
     $c_{i,j}$    & $1 \leq i < j \leq n$                 & $v_i$ is contracted into $v_j$ \\
     $o_{i,j}$    & $1 \leq i < j \leq n$                 & $v_i \prec v_j$ \\
     $p_{i,j}$    & $1 \leq i < j \leq n$                 & $p_i = v_j$ \\
     $r_{i,j,k}$  & $1 \leq i,j \leq n$ and $j < k \leq n$     & $v_j v_k \in E(H_{\pos{v_i}})$ after eliminating $v_i$\\
    \bottomrule
  \end{tabular}
\end{table*}
In our first encoding, we use a relative ordering of the vertices, as
used in the treewidth encoding by Samer and~Veith~\cite{SamerVeith09}:
instead of encoding $\pos{v_i}$ directly, we encode for vertices
$v_i, v_j \in V(G)$, whether $\pos{v_i} < \pos{v_j}$ or not.
Table~\ref{tab:var-relative} shows the variables utilized in the
encoding.  For the ordering, we use $\binom{n}{2}$ variables $o_{i,j}$
with $i < j$, where $o_{i,j}$ is true if and only if $v_i \prec v_j$.
We subsequently use the shorthand $o^*_{i,j}$ where $o^*_{i,j}$ is
$o_{i,j}$ if $i < j$ and $\neg o_{j,i}$ if $i > j$.  We encode the
semantics by enforcing transitivity: for mutually distinct $i,j,k$ we
add the clauses
\[
\neg o^*_{i,j} \vee \neg o^*_{j,k} \vee o^*_{i,k}.
\]

Next, we encode the contraction tree $T$.  In view of
Theorem~\ref{the:leo}, we can assume that when $p_i$ is the parent of~$p_j$ in $T$, then $i<j$ (Condition 1), and $v_n$ is the root of $T$
(Condition 2).  Hence, we can use $\binom{n}{2}$ variables $p_{i,j}$
with $i < j$, where $p_{i,j}$ is true if and only if $p_i = v_j$.
We encode that every vertex, except the root, has exactly one
parent. For that, we utilize at-least-one constraints by
  adding  for each $i < n$ 
the clause $\bigvee_{i < j} p_{i, j}$ and
at-most-one constraints by adding for mutually distinct $i,j,k$ the clause
$\neg p_{i,j} \vee \neg p_{i, k}$.
Additionally, we ensure that
$v_i \prec v_j$ holds between a vertex~$v_i$ and its parent~$v_j$, by
adding for $i<j$ the clauses
\[\neg p_{i,j} \vee o^*_{i,j}.\]

So far we have encoded $\prec$ and $T$.  Next, we encode the elimination
sequence $H_0, \dots, H_n$ with two additional sets of variables.  We
take $n \binom{n}{2}$ variables $r_{i,j,k}$ with $j < k$, where
$r_{i,j,k}$ is true if and only if after eliminating $v_i$ it holds
that $v_j v_k \in E(H_{\pos{v_i}})$.  We also use $\binom{n}{2}$
auxiliary variables $a_{i,j}$ with $i < j$, where $a_{i,j}$ is true if
and only if there exists a $k$ such that $v_iv_j \in E(H_k)$.
We use shorthands $a^*$ and $r^*$ which are defined analogously to~$o^*$.

We encode the semantics of $a$ by adding, for all mutually distinct $i,j,k$, $i<j$, the clause
\[\neg o^*_{i,j} \vee \neg o^*_{i, k} \vee \neg r^*_{i, j, k} \vee a^*_{j, k}.\]
Furthermore, we encode the semantics of $r$ by encoding
Subsets~\eqref{eq:keep}--\eqref{eq:create-target} of $E(H_i)$
according to the definition given in Section~\ref{section:twin-width}.
Subsets~\eqref{eq:create-source} and~\eqref{eq:create-target} are encoded
by adding for $i < j$ and
$v_k \in (N_G(v_i)  \symdiff   N_G(v_j)) \setminus \{v_i, v_j\}$ the
clause
\[\neg p_{i,j} \vee \neg o^*_{i, k} \vee r^*_{i,j,k}.\]
Further, Subset~\eqref{eq:transfer} is encoded by adding, for mutually
distinct $i,j,k$, $i < j$, the clause
\[\neg p_{i, j} \vee \neg o^*_{i,k} \vee \neg a^*_{i,k} \vee
r^*_{i,j,k}.\]
Finally, we encode Subset~\eqref{eq:keep} by adding for mutually
distinct $i,j,k,m$, $k < m$ the clause
\[\neg o^*_{i, j} \vee \neg o^*_{j, k} \vee \neg o^*_{j, m} \vee \neg
r^*_{i,k,m} \vee r^*_{j,k,m}.\]
The $O(n^4)$  clauses required to
encode the Subset~\eqref{eq:keep}
dominate the size of the encoding.  Unfortunately, this is
unavoidable: without knowing $\pos{.}$, we have $O(n^2)$ possible
orderings of $v_i,v_j$, and for each such ordering we have $O(n^2)$
possible edges $v_k v_m$.

We enforce the upper bound $d$ by using \emph{cardinality
  constraints}: sets of clauses that encode the less-than constraints
with the help of auxiliary variables.  For each pair $v_i, v_j$ of
vertices, we limit the set
$\SB r^*_{i,j,k}\SM 1 \leq i,j,k \leq n \SE$ to at most $d$ true
values.  Therefore, $v_j$ has at most $d$ neighbors in
$H_{\pos{v_i}}$.  We achieve this by using the \emph{totalizer}
cardinality constraints, as they perform well with our
encoding~\cite{BailleuxEtAl03,RubenEtAl14}.

Since the construction of $F(G, d)$ closely follows the definitions
given in Section~\ref{section:twin-width}, we have the following
result.
\begin{theorem}
  Given a graph $G$ with $n$ vertices and an integer~$d$,
  we can construct in time polynomial in $n + d$ a propositional formula
  $F(G, d)$ which is satisfiable if and only if
  $\tww(G)\leq d$.
\end{theorem}

\subsection{Absolute Encoding}

We can reduce the number of clauses from $O(n^4)$ to $O(n^3)$ by
directly encoding the absolute position of each vertex in $\prec$.
We first give the general idea behind the adapted encoding and then
compare the two encodings.

We use $n (n-1)$ variables $o'_{i,j}$, where
$o'_{i,j}$ is true if and only if $\pos{v_j}=i$.  We encode the
semantics of these variables by assigning each vertex exactly one
position that is unique among all vertices.  With this modificantion, the
indices $i,j,k,m$ refer to positions $\pos{v_i}$, $\pos{v_j}$,
$\pos{v_k}$, $\pos{v_m}$, respectively, rather than the indices of 
$v_i, v_j, v_k, v_m$.  Therefore, the semantics of $r_{i,j,k}$ changes,
and $r_{i,j,k}$ is true if and only if there exists an edge $uv \in E(H_i)$
such that $j = \pos{u}$ and $k =\pos{v}$.

The main advantage of this modification is that $E(H_i)$ can be
succinctly expressed as $\neg r_{i-1,j,k} \vee r_{i,j,k}$, for
$i > 1$.  We also need fewer variables for $r$: since the vertex at
position $i$ is eliminated before the vertex at position $j$, for
$r_{i,j,k}$ it suffices to use indices in the range $i < j < k$.
Finally, we only need to consider the graphs $H_1,\dots,H_{n-d}$, as a
graph with $d$ vertices cannot have a twin-width higher than $d$.
This significantly reduces the number of variables and clauses.

\subsection{Comparison}

The absolute encoding's reduced size in comparison to the relative
encoding comes with the prize of making it more intricate to encode
the various required properties. Most obviously, the encoding of the
ordering with the  variables $o_{i,j}'$ is more complex than the
encoding of the ordering with the  variables~$o_{i,j}$.  Even
more impeding is the impossibility of succinctly encoding that the
parent of a vertex is lexicographically larger than the vertex itself.
Without this, we are left with many symmetries in the absolute encoding,
which unnecessarily increases the search space.   Encoding the edges is also
considerably more intricate in the absolute encoding: since we do not
know the value of $\pos{v_i}$ in advance, we have to encode for each
edge $v_iv_j \in E(G)$ that there is an edge from $\pos{v_i}$ to
$\pos{v_j}$, which requires $n (n-1)$ variables and $O(n^3)$ clauses.

To illustrate the encoding size, take as an example \Paley{73}, a graph with 73
vertices and 1314 edges and twin-width 36.  The relative encoding requires 30 million
clauses and 2.5 million variables, while the absolute encoding
requires only 2.5 million clauses and 0.3 million variables.

The aforementioned disadvantages of the absolute encoding severely
hinders its performance.  \Paley{73}'s twin-width is found by the
relative encoding within three hours, while the absolute encoding
fails to find the optimal result for a 13-vertex graph within four
hours.

While ill-suited for finding the optimal twin-width, the small size of
the absolute encoding makes it useful for computing upper bounds on
the twin-width of larger graphs.  The last unsatisfiable case
$F(G,\tww(G)-1)$ and the first satisfiable case $F(G,\tww(G))$ usually
take an order of magnitude longer to solve than other cases.
Particularly for $F(G,\tww(G)+i), i=1,2,\dots$ the solving time
decreases quickly. Thus, the absolute encoding can compute upper
bounds on the twin-width for graphs that are too large for the
relative encoding.

\section{Lower and Upper Bounds}\label{sec:bounds}
In this section, we describe a simple approach for deriving lower and
upper bounds for the twin-width of graphs. We use
these bounds for limiting the range for $d$ when running the SAT solver on
$F(G,d)$.

We first discuss the lower bound.
Let $r$ be a positive integer and $G$ a
graph with at least~$r$ vertices. We define the \emph{lower bound
  $\lb_r$ of order $r$} for $\tww(G)$ as the maximum degree of the
first $r+1$ graphs $H_0,\dots,H_{r-1}$ of any elimination sequence for
$G$. In particular, for $r=1$ we have 
\[
  \lb_1(G)=\min_{u,v\in V(G), u\neq v} \Card{N_G(u) \symdiff N_G(v)}.
\]
Clearly, $\lb_1(G) \leq \lb_2(G) \leq \dots \leq \lb_{n}(G) =\tww(G)$.
If~$r$ is a constant, then $\lb_r(G)$ can be computed in polynomial
time.

For obtaining an upper bound on the twin-width of a given graph $G$,
we propose a simple greedy algorithm.  The algorithm computes an
elimination ordering $\prec$ and a contraction tree $T$ step-by-step,
greedily choosing the next vertex~$v_i$ in the ordering. Assume we
have already computed the first $i$ vertices of the elimination
ordering $v_1,\dots,v_{i-1}$ and the corresponding sequence of graphs
$H_0,\dots,H_{i-1}$ with $V(H_{i-1})=\{v_i,\dots,v_n\}$. We choose the
next vertex $v_i \in V(H_{i-1})$ and the corresponding parent
$p_i\in \{v_{i+1},\dots,v_n\}$, $p_i < v_i$ in the lexicographic
ordering of the vertices, such that the degree of $p_i$ in $H_i$ is
minimized; in case of a tie, we take the lexicographically minimal
pair $(v_i,p_i)$. We add the edge $v_ip_i$ to the contraction
tree. The width of the resulting twin-width decomposition $(T,\Prec)$
gives the upper bound $\ub$ on the twin-width of~$G$. Our
implementation of the greedy heuristic uses caching to avoid computing
the degree of potential pairs $(v_i,p_i)$ over and over again.

\section{Experiments}
We computed the twin-width of several graphs using the relative encoding\footnote{Source code can be found at \url{https://github.com/ASchidler/twin_width}.
The results can be found at \url{https://doi.org/10.5281/zenodo.5564192}.}.
We implemented and run the encoding using Python~3.8.0
and PySAT~1.6.0\footnote{\url{https://pysathq.github.io}}.
As the SAT solver, we used  Cadical\footnote{\url{http://fmv.jku.at/cadical/}},
as it worked slightly better with the encoding than the other
solvers provided by PySAT\@. We used a computer
with an Intel Core i5-9600KF CPU running at 3.70~GHz, 32 GB RAM
and Ubuntu 20.04.

\subsection{Named Graphs}\label{section:named}

We computed the twin-width of several named graphs which are
well-known from the literature~\cite{MathWorld}.  The names of the
graphs either reflect their topology or their discoverer.  For most of
the considered graphs, the twin-width was not
known. Table~\ref{tab:famous} provides an overview of our results,
including lower and upper bounds as described in
Section~\ref{sec:bounds}. Preprocessing has no effect on the named
graphs, which all turned out to be prime (as one would expect, as
these graphs often provide a smallest example or counterexample for a
combinatorial property).

\begin{table*}[htb!]
    \centering
    \caption{Results for famous named graphs. For all graph not marked
      with *, the twin-width could be computed in at most five
      seconds.  $\lb_1$ gives the lower bound of order 1, $\ub$ gives
      the width of an elimination ordering computed by the greedy
      algorithm of Section~\ref{sec:bounds}. }
         \label{tab:famous}
\cv{\renewcommand{\X}{\hspace{10mm}}}
\av{\renewcommand{\X}{\hspace{8mm}}}
    \begin{tabular}{@{}l@{\X}r@{\X}r@{\X}r@{\X}r@{\X}r@{\X}r@{\X}r@{}}
      \toprule 
        Graph&$\Card{V}$&$\Card{E}$&$\lb_1$&$\tww$& $\ub$ &Variables&Clauses\\
      \midrule 
      Brinkmann&21&42&6&6&6&34526&150770\\
Chvátal&12&24&2&3&5&5611&18288\\
Clebsch&16&40&6&6&8&15510&64517\\
Desargues&20&30&4&4&5&28383&132636\\
Dodecahedron&20&30&4&4&4&26863&126244\\
Dürer&12&18&2&3&4&5347&18602\\
Errera&17&45&4&5&6&17720&75895\\
FlowerSnark&20&30&4&4&4&28383&119176\\
Folkman&20&40&2&3&3&10311&35761\\
Franklin&12&18&2&2&4&5347&16354\\
Frucht&12&18&2&3&3&5083&17573\\
Goldner&11&27&2&2&4&4067&11813\\
Grid $6\times 8$*&48&82&2&3&4&396751&3493676\\
Grötzsch&11&20&2&3&5&4287&13910\\
Herschel&11&18&2&2&4&4067&13590\\
Hoffman&16&32&2&4&5&14070&58051\\
Holt&27&54&6&6&7&79513&405925\\
Kittell&23&63&4&5&6&46161&171811\\
McGee&24&36&4&4&5&50087&238494\\
Moser&7&11&2&2&2&252&502\\
Nauru&24&36&4&4&5&50087&239051\\
\Paley{73}*&73&1314&36&36&64&2530300&21107035\\
Pappus&18&27&4&4&5&20399&89670\\
Peterson&10&15&4&4&4&3009&9388\\
Poussin&15&39&3&4&5&11571&31049\\
Robertson&19&38&6&6&6&25369&114592\\
Rook $6 \times 6$*&36&180&10&10&12&216499&1236368\\
Shrikhande&16&48&6&6&8&15510&64431\\
Sousselier&16&27&4&4&5&14070&51414\\
Tietze&12&18&2&4&4&5347&18628\\
Wagner&8&12&2&2&2&1418&3909\\
      \bottomrule
    \end{tabular}
       \end{table*}

 \begin{table*}[tbh]
      \caption{Results for Paley graphs.
     The twin-width agrees with the  lower bound of $(\Card{V}-1)/2$.
   Time shows the number of seconds it took to solve the SAT instance.}
   \label{tab:paley}
   \centering
   \renewcommand{\X}{\hspace{10mm}}
   \begin{tabular}{@{}l@{\X}r@{\X}r@{\X}r@{\X}r@{\X}r@{\X}r@{}}
     \toprule
     Name & $\Card{V}$ & $\Card{E}$ & ~~\tww & Variables & Clauses & Time [s] \\
     \midrule
\Paley{09}&9&18&4&2080&6176&$<$1\\
\Paley{13}&13&39&6&7962&29205&$<$1\\
\Paley{17}&17&68&8&19352&84652&$<$1\\
\Paley{25}&25&150&12&73948&408838&2.8\\
\Paley{29}&29&203&14&120406&715814&7.6\\
\Paley{37}&37&333&18&272166&1916941&21.4\\
\Paley{41}&41&410&20&384324&2030173&63.6\\
\Paley{49}&49&588&24&692352&4513244&210.2\\
\Paley{53}&53&689&26&893986&6282603&364.3\\
\Paley{61}&61&915&30&1406886&11437512&2396.8\\
\Paley{73}&73&1314&36&2530300&21107035&9934.3\\
     \bottomrule
   \end{tabular}

\end{table*}
Interestingly, the lower bound $\lb_1$ often coincides with the exact
twin-width.  One possible explanation is the high level of symmetry in
many of the graphs.  A particularly interesting class of symmetric
graphs are the \emph{strongly regular graphs}: these graphs are
usually 
parameterized by the tuple $(n, k, \lambda, \mu)$, where $n$ is the number of
vertices, $k$ is the degree of each vertex, and every pair of vertices
has either $\lambda$ common neighbors if they are adjacent, or share
$\mu$ neighbors otherwise.  For a strongly regular graph $G$
  with parameters $(n, k, \lambda, \mu)$ we can immediately determine
  the lower bound of order 1
  \[  \lb_1(G)=\min\{2(k-\mu), 2(k-\lambda-1)\}.\]
Examples of strongly regular graphs in Table~\ref{tab:famous} are \emph{Clebsch}~$(16,5,0,2)$, \emph{Peterson}~$(10,3,0,1)$,
\emph{Rook $n \times n$}~$(n^2, 2n{-}2, n{-}2, 2)$, and \emph{Shrikhande}~$(16,6,2,2)$.
%For all strongly regular graphs in Table~\ref{tab:famous}, the lower bound matches the twin-width.
A family of strongly regular graphs, the \emph{Paley graphs},
stick out due to their high twin-width in relation to their size.
%Paley graphs are named after the mathematician Raymond Paley (1907--1933)~\cite{BrouwerCohenNeumaier89}
For every prime power $n$ such that $n \equiv 1 \pmod{4}$,
the Paley graph on~$n$ vertices ($\Paley{n}$) is defined
and is strongly regular with parameters $k{=}(n{-}1)/2$, $\lambda{=}(n{-}5)/4$, $\mu{=}(n{-}1)/4$.
Further, Paley graphs are \emph{self-complementary}, i.e., $\Paley{n}$ and
$\overline{\Paley{n}}$ are isomorphic~\cite{GodsilRoyle01}.
With our relative SAT encoding, we could verify  that for Paley graphs with up
to 73 vertices, the lower bound of order 1
gives the exact twin-width, see Table~\ref{tab:paley}. We hope that by
analyzing the twin-width decomposition  provided by our encoding, one can
verify that $\tww(\Paley{n})=(n-1)/2$ holds in general.

Table~\ref{tab:paley} also highlights the quickly increasing size of
our relative encoding.
Despite the size, the solving times are comparatively short.
Although the encoding can compute the twin-width for $\Paley{73}$,
it often starts struggling for general graphs with more than 40 vertices.
This suggests that some graphs are considerably harder for our encoding than others,
independent of their size.

Two-dimensional \emph{grid graphs} are interesting for
  twin-width. They are known to have unbounded treewidth and
  clique-width, but it is easy to see that their twin-width is at most
  4~\cite{Bonet2020}. Interestingly, with our relative encoding, we found
that  smaller grid graphs, of size up to $8 \times 6$,  do have  twin-width~3.   We see it as an interesting challenge to determine the exact
  twin-width of all square grids. The width-3 decompositions that we
  found with our encodings do not suggest any obvious general pattern
  that could be generalized to all grid graphs, hence we still expect
  that at a certain size the width switches from 3 to~4.

\begin{figure}[tbh!]
  \centering
 \includegraphics[scale=0.5]{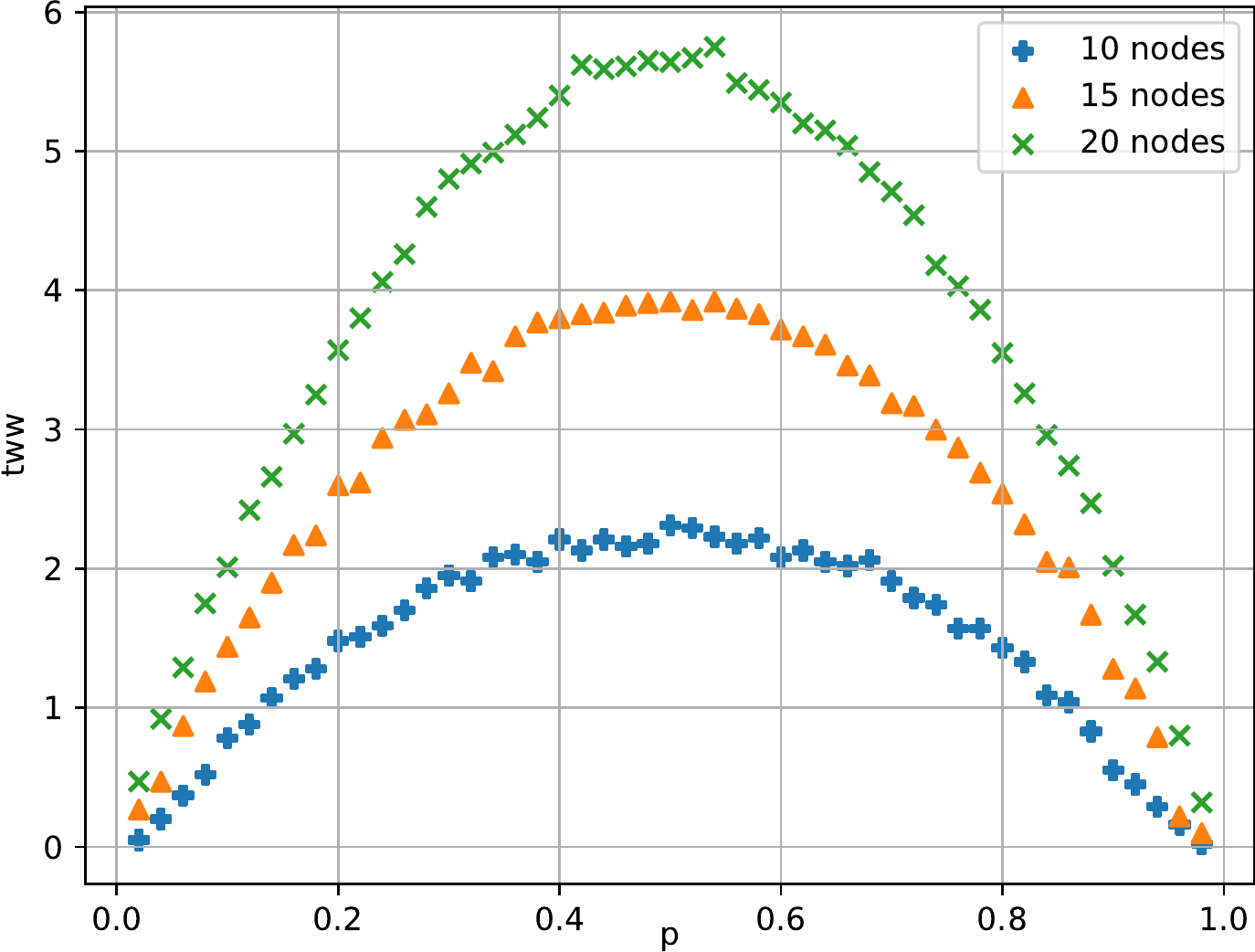}
  \caption{Twin-width for randomly generated graphs: each edge exists with probability $p$.
  Each point is the average over 100 graphs.}
  \label{fig:random}
\end{figure}
\subsection{Random Graphs}

We tested the twin-width on randomly generated graphs.
For this purpose, we created Erd\H{o}s-R\'{e}ny graphs $G(n, p)$,
where $\Card{V(G)}=n \in \{10, 15, 20\}$ and 
each  edge exists with probability~$p$, where $p$ takes values between $0$ to $1$ in $0.02$ increments.

The results in Figure~\ref{fig:random} show that the twin-width increases quickly with increasing graph size.
Furthermore, the vertical distance between the peaks is similar.
The symmetric shape is expected due to Fact~\ref{fact:complement}.

Many of the graphs can be simplified using the preprocessing discussed in Section~\ref{section:preprocessing}.

\subsection{The Twin-Width Numbers}

For every $d>0$, let $\tww_d$ be smallest integer such that there
exists a graph with $\tww_d$ many vertices of twin-width~$d$. We call $\tww_d$
the \emph{$d$-th twin-width} number. In contrast to other width measures like
treewidth, where similar numbers are easy to compute (the $d$-th
treewidth number is $d+1$), no uniform method is known for computing
the twin-width numbers. The situation is similar for clique-width,
where no uniform method is known either; Heule and Szeider~\cite{HeuleSzeider15}
computed the first few clique-width numbers.% with a SAT encoding.

The computation of twin-width numbers provides a challenge for any
exact method, as the search space grows quickly with each increment
of~$d$. However, with our encodings, run on prime graphs generated by
Nauty\footnote{\url{http://cs.anu.edu.au/people/bdm/}}~\cite{McKay2014}, we were able to identify the first few
twin-width numbers and give tight bounds for further ones. 
\begin{proposition}\label{prop:tww-numbers}
  The sequence of twin-width numbers starts with $4,5,8,9$; the fifth
  twin-width number is $11$ or $12$, the sixth twin-width number is at
  most~$13$.
\end{proposition}
For computing the twin-width numbers, we only need to consider graphs
$G$ with $\Card{E(G)}\leq \binom{n}{2}/2$, as by
Fact~\ref{fact:complement}, $\Card{E(G)}> \binom{n}{2}/2$ implies
$\Card{E(\overline{G})}\leq \binom{n}{2}/2$.  Further, according to
Theorem~\ref{the:moddec}, we only need to consider prime graphs. In
particular, since every prime graph $G$ and its complement graph
$\overline{G}$ are connected, we only need to consider connected
graphs. The results are shown in Table~\ref{tab:nauty}.

The preprocessing described in Section~\ref{section:preprocessing} can
be used for all graphs that are not prime.  We can see in
Table~\ref{tab:nauty} that there are many connected graphs that are
not prime, and thereby eligible for preprocessing.

Interestingly, for the first, second, and fourth twin-width number
$\tww_d$, there is a unique graph, up to isomorphism, with $\tww_d$ many
vertices and twin-width~$d$. For the third twin-width number,
there are five such graphs: $G_{8,3,i}, i = 1,\dots,5$.
$G_{8,3,3}$ is self-complementary; the other four form two complementary pairs.
In Figure~\ref{fig:twwn}, we display 
these graphs, together with an optimal $d$\hy sequence,
showing only one graph from each complementary pair.

The unique graph certifying $\tww_1=4$ is the path on~4 vertices
($P_4$). The unique graph certifying $\tww_1=5$ is the cycle on five
vertices ($C_5$).  The unique graph certifying $\tww_4=9$ is the graph
$\Paley{9}$ (see Section~\ref{section:named}).  In fact,
$C_5= \Paley{5}$, so also $\tww_2$ is certified by a Paley graph.
Further, if we remove any vertex from $\Paley{9}$, we obtain
$G_{8,3,3}$.  Similarly, we obtain $P_4$ by removing a vertex from
$\Paley{5}$.  Therefore, Paley graphs are related with all of the
first four twin-width numbers.
We could establish
with our method that among all graphs
with 10 vertices, there is no graph of twin-width 5, hence
$\tww_5\geq 11$. We could not check all graphs with 11 vertices, as
there are too many.
$\Paley{13}$ shows that $\tww_6\leq 13$. By deleting any single vertex
from $\Paley{13}$, its twin-width drops to 5. This implies that
$\tww_5\leq 12$, and so $11 \leq \tww_5 \leq 12$ as stated in Proposition~\ref{prop:tww-numbers}.

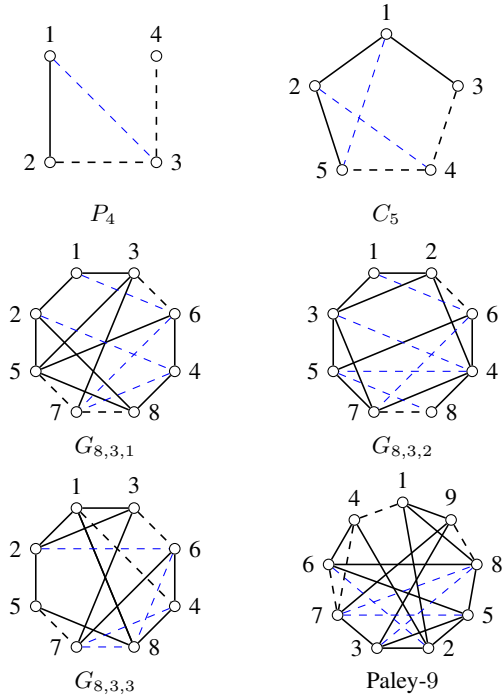
\begin{figure}[tbh]
  \centering
 
  \tikzset{
    lead/.style={
        execute at end picture={
            \coordinate (lower right) at (current bounding box.south east);
            \coordinate (upper left) at (current bounding box.north west);
        }
    },
    follow/.style={
        execute at end picture={
            \pgfresetboundingbox
            \path (upper left) rectangle (lower right);
        }
    }
}

\begin{tikzpicture}[scale=1.0, lead]%smallest_1_0
  \small
  \draw(0,-1.4) node (name) {$P_{4}$};
  \draw (0,0) node[coordinate] (origin) {}
  +(1.5*360/4: 1cm) node[var, label={above:1}] (1) {}
  +(2.5*360/4: 1cm) node[var, label={left:2}] (2) {}
  +(3.5*360/4: 1cm) node[var, label={right:3}] (3) {}
  +(4.5*360/4: 1cm) node[var, label={above:4}] (4) {};
  
  \draw [style=B] (2) to (1);
  \draw [style=BD] (2) to (3);
  \draw [style=BD] (3) to (4);
  \draw [style=GD] (1) to (3);
  
\end{tikzpicture}\hspace{12mm}%
\begin{tikzpicture}[lead, scale=1.0]%smallest_2_0
  \small
  \draw(0,-1.4) node (name) {$C_{5}$};
  \draw (0,0) node[coordinate] (origin) {}
  +(90+0*360/5: 1cm) node[var, label={above:1}] (1) {}
  +(90+1*360/5: 1cm) node[var, label={left:2}] (2) {}
  +(90+2*360/5: 1cm) node[var, label={left:5}] (5) {}
  +(90+3*360/5: 1cm) node[var, label={right:4}] (4) {}
  +(90+4*360/5: 1cm) node[var, label={right:3}] (3) {};
  \draw [style=B] (1) to (2);
  \draw [style=B] (1) to (3);
  \draw [style=B] (2) to (5);
  \draw [style=BD] (3) to (4);
  \draw [style=BD] (4) to (5);
  \draw [style=GD] (1) to (5);
  \draw [style=GD] (2) to (4);
\end{tikzpicture}\hfill\\%
\begin{tikzpicture}[lead,scale=1.0]%smallest_3_0
  \small
  \draw(0,-1.4) node (name) {$G_
    {8,3,1}$};
  \draw (0,0) node[coordinate] (origin) {}
  +(2.5*360/8: 1cm) node[var, label={above:1}] (1) {}
  +(3.5*360/8: 1cm) node[var, label={left:2}] (2) {}
  +(4.5*360/8: 1cm) node[var, label={left:5}] (5) {}
  +(5.5*360/8: 1cm) node[var, label={left:7}] (7) {}
  +(6.5*360/8: 1cm) node[var, label={right:8}] (8) {}
  +(7.5*360/8: 1cm) node[var, label={right:4}] (4) {}
  +(8.5*360/8: 1cm) node[var, label={right:6}] (6) {}
  +(9.5*360/8: 1cm) node[var, label={above:3}] (3) {};
  
  \draw [style=B] (1) to (2);
  \draw [style=B] (1) to (3);
  \draw [style=B] (2) to (8);
  \draw [style=B] (2) to (5);
  \draw [style=BD] (3) to (6);
  \draw [style=B] (3) to (7);
  \draw [style=B] (3) to (5);
  \draw [style=B] (6) to (4);
  \draw [style=B] (6) to (5);
  \draw [style=B] (4) to (8);
  \draw [style=BD] (5) to (7);
  \draw [style=B] (5) to (8);
  \draw [style=BD] (7) to (8);
  \draw [style=GD] (1) to (6);
  \draw [style=GD] (2) to (4);
  \draw [style=GD] (6) to (7);
  \draw [style=GD] (4) to (7);
\end{tikzpicture}\hspace{12mm}%
\begin{tikzpicture}[follow, scale=1.0]%smallest_3_1
  \small
  \draw(0,-1.4) node (name) {$G_{8,3,2}$};
  \draw (0,0) node[coordinate] (origin) {}
  +(2.5*360/8: 1cm) node[var, label={above:1}] (1) {}
  +(3.5*360/8: 1cm) node[var, label={left:3}] (3) {}
  +(4.5*360/8: 1cm) node[var, label={left:5}] (5) {}
  +(5.5*360/8: 1cm) node[var, label={left:7}] (7) {}
  +(6.5*360/8: 1cm) node[var, label={right:8}] (8) {}
  +(7.5*360/8: 1cm) node[var, label={right:4}] (4) {}
  +(8.5*360/8: 1cm) node[var, label={right:6}] (6) {}
  +(9.5*360/8: 1cm) node[var, label={above:2}] (2) {};

  \draw [style=B] (2) to (1);
  \draw [style=BD] (2) to (6);
  \draw [style=B] (2) to (3);
  \draw [style=B] (2) to (4);
  \draw [style=B] (1) to (3);
  \draw [style=B] (6) to (5);
  \draw [style=B] (6) to (4);
  \draw [style=B] (3) to (5);
  \draw [style=B] (3) to (7);
  \draw [style=B] (4) to (8);
  \draw [style=B] (4) to (7);
  \draw [style=B] (5) to (7);
  \draw [style=BD] (7) to (8);
  \draw [style=GD] (1) to (6);
  \draw [style=GD] (6) to (7);
  \draw [style=GD] (3) to (4);
  \draw [style=GD] (4) to (5);
  \draw [style=GD] (5) to (8);

\end{tikzpicture}\hfill\\%
\begin{tikzpicture}[follow, scale=1.0]%smallest_3_2
  \small
  \draw(0,-1.4) node (name) {$G_{8,3,3}$};
  \draw (0,0) node[coordinate] (origin) {}
  +(2.5*360/8: 1cm) node[var, label={above:1}] (1) {}
  +(3.5*360/8: 1cm) node[var, label={left:2}] (2) {}
  +(4.5*360/8: 1cm) node[var, label={left:5}] (5) {}
  +(5.5*360/8: 1cm) node[var, label={left:7}] (7) {}
  +(6.5*360/8: 1cm) node[var, label={right:8}] (8) {}
  +(7.5*360/8: 1cm) node[var, label={right:4}] (4) {}
  +(8.5*360/8: 1cm) node[var, label={right:6}] (6) {}
  +(9.5*360/8: 1cm) node[var, label={above:3}] (3) {};

    \draw [style=B] (1) to (8); 
  \draw [style=B] (2) to (3); 
  \draw [style=B] (2) to (5); 
  \draw [style=B] (2) to (1); 
  \draw [style=BD] (3) to (6);
  \draw [style=B] (3) to (7); 
  \draw [style=B] (3) to (1); 
  \draw [style=B] (5) to (8); 
  \draw [style=BD] (5) to (7);
  \draw [style=B] (1) to (8); 
  \draw [style=BD] (1) to (4);
  \draw [style=B] (6) to (4); 
  \draw [style=B] (6) to (7); 
  \draw [style=B] (4) to (8); 
  \draw [style=GD] (7) to (8);
  \draw [style=GD] (2) to (6);
  \draw [style=GD] (6) to (8);
  \draw [style=GD] (4) to (7);
\end{tikzpicture}\hspace{12mm}%
\begin{tikzpicture}[follow, scale=1.0]%smallest_4_0
  \small
  \draw(0,-1.4) node (name) {$\Paley{9}$};
  \draw (0,0) node[coordinate] (origin) {}
  +(90+0*360/9: 1cm) node[var, label={above:1}] (1) {}
  +(90+1*360/9: 1cm) node[var, label={above:4}] (4) {}
  +(90+2*360/9: 1cm) node[var, label={left:6}] (6) {}
  +(90+3*360/9: 1cm) node[var, label={left:7}] (7) {}
  +(90+4*360/9: 1cm) node[var, label={left:3}] (3) {}
  +(90+5*360/9: 1cm) node[var, label={right:2}] (2) {}
  +(90+6*360/9: 1cm) node[var, label={right:5}] (5) {}
  +(90+7*360/9: 1cm) node[var, label={right:8}] (8) {}
  +(90+8*360/9: 1cm) node[var, label={above:9}] (9) {};

  \draw [style=B] (2) to (5); 
  \draw [style=B] (2) to (1); 
  \draw [style=B] (2) to (3); 
  \draw [style=B] (2) to (4); 
  \draw [style=B] (5) to (6); 
  \draw [style=B] (5) to (8); 
  \draw [style=B] (5) to (3); 
  \draw [style=B] (1) to (9); 
  \draw [style=B] (1) to (8); 
  \draw [style=BD] (1) to (4);
  \draw [style=B] (3) to (9); 
  \draw [style=B] (3) to (7); 
  \draw [style=B] (4) to (6); 
  \draw [style=BD] (4) to (7);
  \draw [style=BD] (6) to (7);
  \draw [style=B] (6) to (8); 
  \draw [style=B] (7) to (9); 
  \draw [style=BD] (8) to (9);
  \draw [style=GD] (2) to (6);
  \draw [style=GD] (5) to (7);
  \draw [style=GD] (3) to (8);
  \draw [style=GD] (7) to (8);

\end{tikzpicture}

%%% Local Variables:
%%% mode: latex
%%% TeX-master: "twwsat"
%%% End:
  \caption{Smallest graphs for given twin-width~$d$. The integer
    vertex labels  give
    a $d$\hy sequence, and the dashed edges give a contraction tree,  as in
    Figure~\ref{fig:wagnerTree}.}
  \label{fig:twwn}
\end{figure}

\begin{table}[htb!]
  \caption{The number of graphs, prime graphs, and prime graphs of a specific twin-width, with a specific number of vertices.}
  \label{tab:nauty}
  \centering
\begin{tabular}{@{}r@{~~}r@{~~~}r@{~~}r@{~~}r@{~~}r@{~~}r@{~~}r@{}}
\toprule
\multicolumn{3}{c}{} & \multicolumn{3}{c}{twin-width}\\
  \cmidrule{4-7}
$\Card{V}$&	 connected&	 prime&		1&	2&	3&	4\\
\midrule
4&	 3&	 1&		1&	0&	0&	0\\
5&	 11&	 4&		3&	1&	0&	0\\
6&	 73&	 26&		16&	10&	0&	0\\
7&	 618&	 260&		90&	170&	0&	0\\
8&	 8573&	 4670&		655&	4010&	5&	0\\
9&	 224875&	 145870&		4488&	137565&	3816&	1\\
10&	 11716571&	 8110356&		30318&	6144756&	1935226&	56\\
\bottomrule
\end{tabular}
\end{table}
\section{Conclusion}
We proposed the first practical approach to computing the exact
twin-width of graphs, utilizing the power of state-of-the-art
SAT-solvers. This allowed us to reveal the twin-width of several
important graphs. Our results provide the first step for showing
general twin-width bounds for infinite graph classes.  For instance,
our data  suggests $\tww(\Paley{n})=(n-1)/2$. Surprisingly, up
to $n=6$, the $n\times n$ grids have twin-width $3$. It would be
interesting to know if and when twin-width 4 is required.  Another
possible application of our results is the construction of gadgets for
showing the theoretical intractability of twin-width computation. Such
intractability is expected~\cite{BonnetEtal(1)20}, but no proof has
yet been found.

The two proposed SAT encodings' different performance
is impressive: the relative encoding benefits from symmetry breaking
and vastly outperforms the more succinct absolute encoding. Although
the relative encoding doesn't explicitly exploit the input graph's
symmetries,
it performs well on some highly symmetric graphs
like $\Paley{73}$.

We hope that our results provide new insights and stimulates further
theoretical investigations on twin-width.  We also hope that our
results provide a first step towards a practical use of twin-width. A
next step would be the implementation and testing of twin-width-based
dynamic programming algorithms like the algorithms for $k$\hy
Independent Set and $k$\hy Dominating Set proposed by Bonnet et
al.~\shortcite{BonnetEtal(3)20}, which are single exponential in the
twin-width.

\bibliographystyle{plainurl}
\bibliography{literature}
\end{document}